\documentclass[%
 aip,
jmp,
 amsmath,amssymb,
preprint,
]{revtex4-2}

\usepackage{graphicx}
\usepackage{dcolumn}
\usepackage{bm}

\usepackage[utf8]{inputenc}
\usepackage[T1]{fontenc}
\usepackage{mathptmx}
\usepackage{etoolbox}
\usepackage{color}

\usepackage{amsthm}
\theoremstyle{plain}
\newtheorem{lemma}{Lemma}
\newtheorem{thm}{Theorem}

\newtheorem{corollary}{Corollary}
\newtheorem{proposition}{Proposition}

\theoremstyle{definition}
\newtheorem{defi}{Definition}
\newtheorem{example}{Example}

\theoremstyle{remark}
\newtheorem{remark}{Remark}

\begin{document}

\preprint{AIP/123-QED}

\title{A Sufficient Criterion for Divisibility of Quantum Channels}
\author{Frederik vom Ende}\email{frederik.vom.ende@fu-berlin.de}
\affiliation{ 
Dahlem Center for Complex Quantum
Systems, Freie Universität Berlin, Arnimallee 14, 14195 Berlin, Germany
}%

\date{\today}

\begin{abstract}
We present a simple, dimension-independent criterion which guarantees that some quantum channel $\Phi$ is divisible, i.e., that there exists a non-trivial factorization $\Phi=\Phi_1\Phi_2$. The idea is to first define an ``elementary'' channel $\Phi_2$ and then to analyze when $\Phi\Phi_2^{-1}$ is completely positive. The sufficient criterion obtained this way---which even yields an explicit factorization of $\Phi$---is that one has to find orthogonal unit vectors $x,x^\perp$ such that $\langle x^\perp|\mathcal K_\Phi\mathcal K_\Phi^\perp|x\rangle=\langle x|\mathcal K_\Phi\mathcal K_\Phi^\perp|x\rangle=\{0\}$ where $\mathcal K_\Phi$ is the Kraus subspace of $\Phi$ and $\mathcal K_\Phi^\perp$ is its orthogonal complement. Of course, using linearity this criterion can be reduced to finitely many equalities. Generically, this division even lowers the Kraus rank which is why repeated application---if possible---results in a factorization of $\Phi$ into in some sense ``simple'' channels. Finally, be aware that our techniques are not limited to the particular elementary channel we chose.
\end{abstract}

\maketitle

\section{Introduction}\label{sec:intro}

The aim of open systems theory is to study the---often inevitable---interaction between a quantum system and its surroundings, as well as the resulting system's dynamics -- but without explicitly keeping track of the environment.
Mathematically, this boils down to describing the evolution of a quantum system via a family of quantum channels
$
\{\Lambda_t\}_{t\geq 0}
$
which act only on the given system.
In the worst case $\Lambda_t$ comes about via a reduced memory kernel master equation \mbox{$\dot\Lambda_t= \int_0^t K_{t-\tau}\cdot\Lambda_\tau\,d\tau$} also known as Nakajima--Zwanzig equation \cite{Haake73,Nakajima58,Zwanzig60}.
However, there are plenty of scenarios where the reduced dynamics admit more structure so computations can be simplified.
Such scenarios are commonly characterized via the notion of 
divisibility\cite{Wolf08b,RHP10,ChruKossRivas11}
where one asks the following question: 
Given some dynamics $\{\Lambda_t\}_{t\geq 0}$ does there exist a family of linear maps $\{V_{t,s}\}_{0\leq s\leq t}$---called propagators---such
that $\Lambda_t=V_{t,s}\Lambda_s$ for all $0\leq s\leq t$?
And if such propagators exist, what further properties do they have? 
The most favorable notion here is
(time-dependent) Markovianity of the dynamics \cite{Wolf08a}
which is known to be equivalent to the propagators $V_{t,s}$ being completely positive.
Sometimes such evolutions are called ``memoryless'' because the propagators 
guarantee that the evolution from time~$s$ to time $t>s$ does not depend on anything that happened prior to time~$s$.
Yet, weaker notions such as $k$- or just $P$-divisibility (i.e., the propagators are not completely positive but just ($k$-)positive) \cite{CM14} as well as Kadison--Schwarz divisibility \cite{CM19} have been investigated in the past.
Importantly, $P$-divisibility is equivalent to lack of information backflow from the system to the environment 
\cite{ChruKossRivas11,CM14}.
For more results regarding divisibility in the dynamical regime we refer to the recent report of Chru{\'s}ci{\'n}ski \cite{Chrus22}.

Departing from dynamics for the moment, the notion of divisibility first came about for (static) quantum channels\cite{Wolf08a}, refer also to Sec.~\ref{sec_prelim} for more detail.
There, given some channel $\Phi$, one asks whether there exists a non-trivial decomposition $\Phi=\Phi_1\Phi_2$ into other quantum channels; hence this can be interpreted as the channel analogue of CP-divisible dynamics.
This concept is most interesting from a semigroup perspective: start from the (quite realistic) assumption that an experimenter can implement only some, but not all possible quantum operations on their system. Yet it could of course be that every---or at least some desired---operation can be implemented via products of what the experimenter can already do. More formally the following questions emerge: 1.~Given some subset $\mathcal S$ of quantum channels what is the semigroup generated by $\mathcal S$, i.e., what is $\langle\mathcal S\rangle_S:=\{\Phi_1\cdot\ldots\cdot\Phi_k:k\in\mathbb N,\Phi_1,\ldots,\Phi_k\in\mathcal S\}$? Similarly, can one find ``small'' sets of channels such that the generated semigroup features all channels?
2.~Given some target channel $\Phi$ can it be ``divided'' into products of elements from $\mathcal S$? (This is of course equivalent to the membership problem $\Phi\in\langle\mathcal S\rangle_S$.)
These questions relate to concepts like universality of gate sets or (approximate) gate compilation\cite{NC10,dawson2005solovay} as these are the corresponding questions in the field of quantum computing.
While there the group structure allows for powerful tools, quantum channels do not form a group but only a semigroup.
This is also why the problem of a universal set of quantum channels has, so far, only been settled in the one-qubit case
\cite[Thm.~5.19]{BGN14}.

With this in mind this paper's core idea is to consider and answer the reverse question: Instead of asking whether some decomposition $\Phi=\Phi_1\Phi_2$ is possible we want to find an operation $\Phi_2$ which is in some sense ``elementary'' and for which one can find simple conditions under which $\Phi_2$ can be ``divided'' from some given channel $\Phi$ (i.e., conditions under which $\Phi\Phi_2^{-1}$ is again a channel).
Originally, this approach was considered for non-negative matrices: there, a divisible matrix is called ``factorizable'' and a lack of divisibility is dubbed ``prime'' \cite{RS74,dCG81,PHS98,vandenHof99}. In the original paper on matrix primes Richman and Schneider considered a two-level transformation $C_\varepsilon$ (based on some parameter $\varepsilon>0$) and, given any non-negative matrix $A$, they characterized when $AC_\varepsilon^{-1}$ is again non-negative for some $\varepsilon$. Their characterization only depends on where the zeros are located in the columns of $A$
\cite[Thm~2.4]{RS74}, cf.~also the end of Sec.~\ref{sec32} below.
Based on this our---quite non-trivial---task will be to lift their idea to quantum channels. The analogue of the two-level matrix $C_\varepsilon$ will be a channel $\Psi_\lambda$ of Kraus rank $2$ (cf.~Sec.~\ref{sec_main} for more detail).
Interestingly, the idea of ``dividing'' by a channel of Kraus rank $2$ has already been mentioned briefly in the original channel divisibility paper \cite[Thm.~11]{Wolf08a} but, until now, this has not been followed up on.
Either way this idea leads to our main result:\medskip

\noindent\textbf{Theorem~\ref{thm_main_div} (Informal).}\textit{
Given any quantum channel $\Phi$ with Kraus operators $\{K_j\}_j$, if there exist
 orthogonal unit vectors $x,x^\perp$ such that 
$\langle x^\perp|K_j^\dagger G_k|x\rangle=0=\langle x|K_j^\dagger G_k|x\rangle$ for all $j,k$---where $\{G_k\}_k$ is any basis of ${\rm span}\{
K_j:j
\}^\perp$---then $\Phi$ is divisible.
If, additionally, $\Phi(|x\rangle\langle x|)\neq\Phi(|x^\perp\rangle\langle x^\perp|)$,
then there exist channels $\Psi_1,\Psi_2$ with $\Phi=\Psi_1\Psi_2$ such that $\Psi_2$ has Kraus rank $2$ and the Kraus rank of $\Psi_1$ is strictly less than the Kraus rank of $\Phi$.
}\medskip

What is more: this is not only an abstract divisibility result but once such $x,x^\perp$ have been found one gets an explicit decomposition $\Phi=(\Phi\Phi_2^{-1})\Phi_2$ (where the channel $\Phi_2$ of Kraus rank $2$ depends, of course, on $x,x^\perp$).
We will prove this by showing that under the above condition the kernel of the Choi matrix of $\Phi$ does not immediately shrink when applying $\Phi_2^{-1}$ from the right. The necessary preliminaries for this can be found in Section~\ref{sec_prelim}.
Then Section~\ref{sec_main} contains a necessary condition for $\Phi\Phi_2^{-1}$ to again be a channel (Sec.~\ref{sec31}),
and our main result is presented in Sec.~\ref{sec32}.
After illustrating our core idea by means of some examples in Section~\ref{sec_ex},
we will conclude and present a brief outlook in Section~\ref{sec_coout}.
In particular, we will argue that this work serves as a ``blueprint'' because the tools presented herein can be adopted to any other elementary channel which one wants to factor out.

\section{Preliminaries: Divisibility and the Kraus subspace}\label{sec_prelim}

First some notation. We write $\mathcal L(\mathbb C^{n\times n},\mathbb C^{m\times m})$ for the collection of all linear maps from $\mathbb C^{n\times n}$ to $\mathbb C^{m\times m}$.
Moreover, $\mathsf{CP}(n,m)\subset \mathcal L(\mathbb C^{n\times n},\mathbb C^{m\times m})$ is the subset of all completely positive maps, that is, those maps $\Phi$ for which $({\rm id}_k\otimes\Phi)(A)\geq 0$ for all $A\in \mathbb C^{k\times k}\otimes\mathbb C^{n\times n}$ positive semi-definite and all $k\in\mathbb N$.
We write $\mathsf{CP}(n):=\mathsf{CP}(n,n)$, as well as $\mathsf{CPTP}(n)$
for those $\Phi\in\mathsf{CP}(n)$ which additionally preserve the trace of any input; the elements of $\mathsf{CPTP}(n)$ are usually called \textit{quantum channels}. Finally, we will use the adjoint representation ${\rm Ad}_X(Y):=XYX^{-1}$ for any $X,Y\in\mathbb C^{n\times n}$ with $X$ invertible. For the unitary group $\mathsf U(n)$ and the special unitary group $\mathsf{SU}(n)$ this of course becomes ${\rm Ad}_U=U(\cdot)U^\dagger$.

An important property of completely positive maps is their decomposition into Kraus operators: $\Phi\in\mathcal L(\mathbb C^{n\times n},\mathbb C^{m\times m})$ is completely positive if and only if there exist $\{K_j\}_{j=1}^\ell\subset\mathbb C^{m\times n}$ such that 
$\Phi\equiv\sum_{j=1}^\ell K_j(\cdot)K_j^\dagger$, cf.~\cite[Thm.~2.22]{Watrous18}{}.
Based on this one also defines $\mathcal K_\Phi:={\rm span}\{K_1,\ldots,K_\ell\}$, called the \textit{Kraus subspace},
an object central to, e.g., the task of distinguishing quantum channels \cite{DGKL16}. Note that $\mathcal K_\Phi$ is well defined because any two sets of Kraus operators to the same completely positive map are linear combinations of each other \cite[Coro.~2.23]{Watrous18}.
Another key tool in studying completely positive maps is the Choi--Jamio\l{}kowski isomorphism
\begin{align*}
\mathsf C:\mathcal L(\mathbb C^{n\times n},\mathbb C^{m\times m})&\to\mathbb C^{n\times n}\otimes\mathbb C^{m\times m}\\
\Phi&\mapsto ({\rm id}_n\otimes\Phi)|\eta\rangle\langle\eta|=\sum_{j,k=1}^n|j\rangle\langle k|\otimes\Phi(|j\rangle\langle k|)
\end{align*}
where $|\eta\rangle:=\sum_{j=1}^n|j\rangle\otimes|j\rangle$ is the unnormalized maximally entangled state.
Obviously $\mathsf C$ is linear, well known to be bijective 
\cite[Sec.~4.4.3]{Heinosaari12}, and its key feature is that it establishes a one-to-one correspondence between completely positive maps and positive semi-definite matrices \cite{Choi75}.
Moreover, the number ${\rm rank}(\mathsf C(\Phi))$ is known as the \textit{Kraus rank} of $\Phi$.

To relate the Choi--Jamio\l{}kowski isomorphism to the Kraus representation we recall the concept of vectorization \cite[Ch.~4.2 ff.]{HJ2} where one defines ${\rm vec}(X):=\sum_{j=1}^n|j\rangle\otimes X|j\rangle\in\mathbb C^n\otimes\mathbb C^m\simeq\mathbb C^{mn}$ for any $X\in\mathbb C^{m\times n}$. Equivalently, ${\rm vec}(X)=({\bf1}_n\otimes X)|\eta\rangle$ so $|\eta\rangle={\rm vec}({\bf1}_n)$.
Expectedly, ${\rm vec}$ is invertible (and even unitary) with inverse ${\rm vec}^{-1}(x)=\sum_j(\langle j|\otimes{\bf1})|x\rangle\langle j|$.
The most well-known property of vectorization is that ${\rm vec}(ABC)=(C^T\otimes A){\rm vec}(B)$ for all ``compatible'' matrices $A,B,C$.
In particular this means
\begin{equation}\label{eq:C_AB}
\mathsf C(A(\cdot)B^\dagger)=\sum_{j,k=1}^n|j\rangle\langle k|\otimes A|j\rangle\langle k|B^\dagger =\Big( \sum_{j=1}^n|j\rangle\otimes A|j\rangle \Big)\Big( \sum_{k=1}^n|k\rangle\otimes B|k\rangle \Big)^\dagger ={\rm vec}(A){\rm vec}(B)^\dagger
\end{equation}
for all $A,B\in\mathbb C^{m\times n}$.
For our purpose it will also be important to know how Choi--Jamio\l{}kowski behaves under multiplication:
\begin{lemma}\label{lemma_choi_prod}
Let $m,m',n,n'\in\mathbb N$, $\Phi\in\mathcal L(\mathbb C^{n\times n},\mathbb C^{m\times m})$, and $\Psi_K\in\mathsf{CP}(m,m')$, $\Psi_L\in\mathsf{CP}(n',n)$ be given. If $\{K_a\}_a\subset\mathbb C^{m'\times m}$, $\{L_b\}_b\subset\mathbb C^{n\times n'}$ are any Kraus operators of $\Psi_K$, $\Psi_L$, respectively, then
$$
\mathsf C(\Psi_K\Phi\Psi_L)=\sum_{a,b} ( L_b^T \otimes K_a ) \mathsf C(\Phi)(L_b^T\otimes K_a)^\dagger \,.
$$
\end{lemma}
\begin{proof}
By \cite[Coro.~2.21]{Watrous18} there exist $\{A_j\}_j,\{B_j\}_j\subset\mathbb C^{m\times n}$
such that $\Phi=\sum_jA_j(\cdot)B_j^\dagger $
so, by assumption, $\Psi_K\Phi\Psi_L=\sum_{a,b,j}K_aA_jL_b(\cdot)(K_aB_jL_b)^\dagger $.
As seen previously in Eq.~\eqref{eq:C_AB} this implies $\mathsf C(\Psi_K\Phi\Psi_L)=\sum_{a,b,j}{\rm vec}(K_aA_jL_b){\rm vec}(K_aB_jL_b)^\dagger $
which in turn concludes the proof:
\begin{align*}
\mathsf C(\Psi_K\Phi\Psi_L)&=\sum_{a,b,j}\Big((L_b^T\otimes K_a) {\rm vec}(A_j) \Big)\Big((L_b^T\otimes K_a) {\rm vec}(B_j) \Big)^\dagger \\
&=\sum_{a,b}(L_b^T\otimes K_a)\Big( \sum_j {\rm vec}(A_j) {\rm vec}(B_j) ^\dagger \Big)(L_b^T\otimes K_a)^\dagger\\
&=\sum_{a,b} ( L_b^T \otimes K_a ) \mathsf C(\Phi)(L_b^T\otimes K_a)^\dagger.\qedhere
\end{align*}
\end{proof}
\noindent As a direct consequence we obtain the familiar transformation rule of the Choi--Jamio\l{}kowski isomorphism under unitary channels\cite[Ch.~10.5]{Bengtsson17}: $\mathsf C({\rm Ad}_{V^\dagger }\cdot \Phi\cdot{\rm Ad}_U)=(U^T\otimes V^\dagger )\mathsf C(\Phi)(U^T\otimes V^\dagger )^\dagger $ for all $U,V\in\mathsf{SU}(n)$ and all $\Phi\in\mathcal L(\mathbb C^{n\times n})$.

What will also be important is the following connection between the kernel of the Choi matrix and the orthogonal complement of the Kraus subspace:
\begin{lemma}\label{lemma_kerC_vecKPhi}
For all $\Phi\in\mathsf{CP}(n)$ it holds that ${\rm ker}(\mathsf C(\Phi))={\rm vec}(\mathcal K_\Phi^\perp)$
where, here and henceforth, $\mathcal K_\Phi^\perp$ is the subspace orthogonal to $\mathcal K_\Phi$ with respect to the Hilbert-Schmidt inner product $\langle A,B\rangle_{\rm HS}:={\rm tr}(A^\dagger B)$.
\end{lemma}
\begin{proof}
In what follows let $\{K_j\}_j\subset\mathbb C^{n\times n}$ be any set of Kraus operators of $\Phi$.

``$\supseteq$'': Let any $Z\in\mathcal K_\Phi^\perp$ be given, that is, $\langle K_j|Z\rangle_{\rm HS}={\rm tr}(K_j^\dagger Z)=0$ for all $j$.
Using that ${\rm vec}$ is unitary this, together with~\eqref{eq:C_AB}, shows
\begin{align*}
\mathsf C(\Phi)|{\rm vec}(Z)\rangle=\sum_j|{\rm vec}(K_j)\rangle\langle{\rm vec}(K_j)|{\rm vec}(Z)\rangle=\sum_j\langle K_j|Z\rangle_{\rm HS}|{\rm vec}(K_j)\rangle
=0\,.
\end{align*}

``$\subseteq$'': Let $z\in{\rm ker}(\mathsf C(\Phi))$. As before one sees $0=\langle z|\mathsf C(\Phi)|z\rangle=\sum_j|\langle K_j|{\rm vec}^{-1}(z)\rangle_{\rm HS}|^2$ so
$\langle K_j|{\rm vec}^{-1}(z)\rangle_{\rm HS}=0$ for all $j$, i.e., ${\rm vec}^{-1}(z)\in\mathcal K_\Phi^\perp$ which is what we had to show.
\end{proof}

The final concept we recap is the one at the core of this work: $\Phi\in\mathsf{CPTP}(n)$ is called \textit{divisible}\cite{Wolf08a} if there exist $\Phi_1,\Phi_2\in \mathsf{CPTP}(n)$ such that 
$\Phi=\Phi_1\Phi_2$ where neither (or both) of $\Phi_1,\Phi_2$ are of the form ${\rm Ad}_U$ with $U\in\mathsf U(n)$.
If $\Phi\in\mathsf{CPTP}(n)$ is not divisible, then it is called \textit{indivisible}.
As an example, the only indivisible qubit channels $\Phi\in\mathsf{CPTP}(2)$ are those which are unital ($\Phi({\bf1})={\bf1}$) and which have Kraus rank $3$, cf.~\cite[Thm.~23]{Wolf08a}.
Moreover, every $\Phi\in\mathsf{CPTP}(n)$ with maximal Kraus rank (i.e., $n^2$) is known to be divisible \cite[Thm.~11]{Wolf08a}:
The idea behind this result is that in the case of maximal Kraus rank all eigenvalues of $\mathsf C(\Phi)$ are positive so given any non-unitary channel $\Psi$, by continuity there exists $\lambda\in(0,1)$ such that the eigenvalues of $\mathsf C(\Phi(\lambda\Psi+(1-\lambda){\rm id})^{-1})$ are positive. 
Thus the decomposition $\Phi=(\Phi(\lambda\Psi+(1-\lambda){\rm id})^{-1})(\lambda\Psi+(1-\lambda){\rm id})$ shows that $\Phi$ is divisible, as claimed.

\section{Main Results}\label{sec_main}
As discussed in the introduction we will port the idea of Richman and Schneider \cite{RS74} to quantum channels by means of the following map:
\begin{defi}
Given any linearly independent unit vectors
$x,y\in\mathbb C^n$ and any $\lambda\in[0,1]$ we define
$
K_\lambda:=
{\bf1}-(1-\sqrt{1-\lambda})|x\rangle\langle x|$, $ L_\lambda:=\sqrt{\lambda}|y\rangle\langle x|
$,
and
\begin{align*}
\Psi_{\lambda,x,y}:\mathbb C^{n\times n}&\to \mathbb C^{n\times n}\\
 \rho &\mapsto K_\lambda \rho K_\lambda^\dagger+L_\lambda \rho L_\lambda^\dagger\,.
 \end{align*}
\end{defi}
\noindent By definition $\Psi_{\lambda,x,y}$---or $\Psi_\lambda$ for short, whenever there is no risk of confusion---is completely positive (with Kraus rank $2$, unless $\lambda=0$) and one readily verifies that $K_\lambda^\dagger K_\lambda+L_\lambda^\dagger L_\lambda={\bf1}$, which is equivalent to $\Psi_\lambda$ being trace-preserving.
This map can be seen as ``elementary'' because for all $z\in({\rm span}\{x\})^\perp$ one has $\Psi_\lambda(|z\rangle\langle z|)=|z\rangle\langle z|$ so the action of $\Psi_\lambda$ is only non-trivial if
$|x\rangle$ is involved:
$\Psi_\lambda(|x\rangle\langle x|)=(1-\lambda)|x\rangle\langle x|+\lambda|y\rangle\langle y|$.\medskip

Let us collect some fundamental properties of $\Psi_\lambda$ in the following lemma:
\begin{lemma}\label{lemma_properties_Psilambda}
Given any linearly independent unit vectors
$x,y\in\mathbb C^n$ the following statements hold.
\begin{itemize}
\item[(i)] $\Psi_\lambda$ is bijective for all $\lambda\in[0,1)$ with inverse
$
\Psi_\lambda^{-1}=K_\lambda^{-1}(\cdot)K_\lambda^{-1}-\,\widehat L_\lambda(\cdot)\widehat L_\lambda^\dagger 
$
where
$$
\widehat L_\lambda:=\;\frac{\sqrt\lambda\langle x|y\rangle}{\sqrt{(1-\lambda)(1-\lambda|\langle x^\perp|y\rangle|^2)}}|x\rangle\langle x|+\frac{\sqrt\lambda\langle x^\perp|y\rangle}{\sqrt{1-\lambda|\langle x^\perp|y\rangle|^2}}|x^\perp\rangle\langle x|\,;
$$
here $x^\perp$ is any unit vector in ${\rm span}\{x,y\}$ which is orthogonal to $x$.
\item[(ii)] If $\langle x|y\rangle=0$, then for all $\lambda,\mu\in[0,1]$ it holds that $\Psi_{\lambda,x,y}\Psi_{\mu,x,y}=\Psi_{\lambda+\mu-\lambda\mu,x,y}$ so, in particular, $[\Psi_{\lambda,x,y},\Psi_{\mu,x,y}]=0$.
\item[(iii)] Assume $\langle x|y\rangle=0$ and let any $\Phi\in\mathsf{CPTP}(n)$ be given. If there exists $\widehat\lambda\in(0,1)$ such that $\Phi\Psi_{\widehat\lambda}^{-1}$ is completely positive, then $\Phi\Psi_{\lambda}^{-1}$ is completely positive for all $\lambda\in[0,\widehat\lambda]$.
\end{itemize}
\end{lemma}
\begin{proof}
(i): 
This follows from the readily-verified identities $K_\lambda^{-1}={\bf1}-(1-(1-\lambda)^{-1/2})|x\rangle\langle x|$, $K_\lambda\widehat L_\lambda=(1-\lambda|\langle x^\perp|y\rangle|^2)^{-1/2}L_\lambda$, $L_\lambda K_\lambda^{-1}=(1-\lambda)^{-1/2}L_\lambda$, and $L_\lambda\widehat L_\lambda=\frac{\sqrt\lambda\langle x|y\rangle}{\sqrt{(1-\lambda)(1-\lambda|\langle x^\perp|y\rangle|^2)}}L_\lambda$.

(ii): Again, the following are identities are key: $K_\lambda K_\mu =K_{\lambda+\mu-\lambda\mu}$, $L_\lambda K_\mu=\sqrt{1-\mu}L_\lambda$, and---because $\langle x|y\rangle=0$---$K_\lambda L_\mu=L_\mu$ as well as $L_\lambda L_\mu=0$.

(iii): Let any $\lambda\in[0,\widehat\lambda]$ be given; thus $\frac{\widehat\lambda-\lambda}{1-\lambda}\in[0,1)$. Then~(ii) implies that $\Psi_{\frac{\widehat\lambda-\lambda}{1-\lambda}}\Psi_\lambda=\Psi_{\widehat\lambda}$.
In particular this shows
$
\Phi\Psi_{\lambda}^{-1}=(\Phi\Psi_{\widehat\lambda}^{-1})\Psi_{\frac{\widehat\lambda-\lambda}{1-\lambda}}
$
meaning $\Phi\Psi_{\lambda}^{-1}$ is completely positive as composition of two completely positive maps.
\end{proof}
\noindent As an aside, the multiplication $(\lambda,\mu)\mapsto\lambda+\mu-\lambda\mu$ also comes up---even in a slightly more general form---for stochastic matrices and in the quantum thermodynamics of a qubit \cite[Sec.~IV]{vomEnde22thermal}.
\begin{remark}
While everything that follows works for general $y$, for simplicity we will only consider the special case $\langle x|y\rangle=0$, e.g., $y=x^\perp$.
Then $\widehat L_\lambda=(1-\lambda)^{-1/2}L_\lambda$ so Lemma~\ref{lemma_properties_Psilambda}~(i) becomes
\begin{equation}\label{eq:Psiinv_xy_orthogonal}
\Psi_\lambda^{-1}=K_\lambda^{-1}(\cdot)K_\lambda^{-1}-\frac{1}{1-\lambda}\,L_\lambda(\cdot)L_\lambda^\dagger\,.
\end{equation}
We expect most of our results to hold for general $y$; however, proving these general statements would be more involved and not necessary for what we are trying to demonstrate in this work.
For example, Lemma~\ref{lemma_properties_Psilambda}~(ii) is false if $\langle x|y\rangle\neq 0$, yet Lemma~\ref{lemma_properties_Psilambda}~(iii) probably still holds although one would have to prove it in a different way.
\end{remark}
The question which will lead us to our main result is the following:
Given some $\Phi\in\mathsf{CPTP}(n)$, under what conditions do there exist $\lambda\in(0,1)$ as well as orthogonal unit vectors $x,x^\perp$ such that $\Phi\Psi_\lambda^{-1}\in\mathsf{CPTP}(n)$?
Note that this problem has to depend on $\Phi$ because $\Psi_\lambda^{-1}$ on its own is never completely positive (as $\Psi_\lambda$ is not unitary, unless of course $\lambda=0$), cf.~\cite[Prop.~4.31]{Heinosaari12} or \cite[Ch.~2, Coro.~3.2]{Davies76}.
A first simple observation to make is that this problem is ``invariant'' under the usual unitary degree of freedom:

\begin{lemma}\label{lemma_div_unitary_equiv}
Let $\Phi\in\mathsf{CPTP}(n)$, $U,V\in\mathsf U(n)$ be given and
define $\Phi_{U,V}:={\rm Ad}_{V^\dagger }\cdot \Phi\cdot{\rm Ad}_U$.
The following statements are equivalent.
\begin{itemize}
\item[(i)] There exist $\lambda\in(0,1)$ and orthogonal unit vectors $x,x^\perp$ such that $\Phi\Psi_\lambda^{-1}\in\mathsf{CPTP}(n)$.
\item[(ii)] There exist $\lambda\in(0,1)$ and orthogonal unit vectors $x,x^\perp$ such that $\Phi_{U,V}\Psi_\lambda^{-1}\in\mathsf{CPTP}(n)$.
\end{itemize}
\end{lemma}
\begin{proof}
``(i) $\Rightarrow$ (ii)'': Define $y:=U^\dagger x$, $y^\perp:=U^\dagger x^\perp$, as well as $\widehat\Psi_\lambda$ via the two Kraus operators
$
{\bf1}-(1-\sqrt{1-\lambda})|y\rangle\langle y|$, $ \sqrt{\lambda}|y^\perp\rangle\langle y|
$. Then  $\widehat\Psi_\lambda={\rm Ad}_{U^\dagger }\cdot \Psi_\lambda\cdot{\rm Ad}_U$
and
$y,y^\perp$ are orthogonal unit vectors, thus
$
\Phi_{U,V}\widehat\Psi_\lambda^{-1}={\rm Ad}_{V^\dagger }\cdot \Phi\cdot{\rm Ad}_U\cdot {\rm Ad}_{U^\dagger }\cdot \Psi_\lambda^{-1}\cdot{\rm Ad}_U={\rm Ad}_{V^\dagger }\cdot (\Phi\cdot \Psi_\lambda^{-1})\cdot{\rm Ad}_U\in\mathsf{CPTP}(n)\,,
$
as desired.
Also ``(ii) $\Rightarrow$ (i)'' is shown analogously (replace $U,V$ by $U^\dagger ,V^\dagger $).
\end{proof}
\subsection{Necessary criterion for complete positivity of $\Phi\Psi_\lambda^{-1}$}\label{sec31}

Because the rest of this paper will be about when $\Phi\Psi_\lambda^{-1}$ is
completely positive we need an analytical expression for its Choi matrix.
This is where we combine the explicit form of $\Psi_\lambda^{-1}$ from Eq.~\eqref{eq:Psiinv_xy_orthogonal} with the
identity from Lemma~\ref{lemma_choi_prod}; doing so yields
\begin{equation}\label{eq:C_Phi_Psiinv}
\mathsf C(\Phi\Psi_\lambda^{-1})=((K_\lambda^{-1})^T\otimes{\bf1})\mathsf C(\Phi)((K_\lambda^{-1})^T\otimes{\bf1})^\dagger -\frac{1}{1-\lambda}(L_\lambda^T\otimes{\bf1})\mathsf C(\Phi)(L_\lambda^T\otimes{\bf1})^\dagger \,.
\end{equation}
Based on this we first derive a necessary criterion:
\begin{proposition}\label{prop_necess_division}
Let $\Phi\in\mathsf{CPTP}(n)$ and orthogonal unit vectors $x,x^\perp\in\mathbb C^n$ be given. If $\Phi\Psi_\lambda^{-1}$ is completely positive for some $\lambda\in(0,1)$, then one of the following equivalent statements holds.
\begin{itemize}
\item[(i)]
$\iota^\dagger (|x^\perp\rangle\langle x|^T\otimes{\bf1})\mathsf C(\Phi)(|x^\perp\rangle\langle x|^T\otimes{\bf1})^\dagger \iota=0$ with $\iota:\mathbb C^{n^2-{\rm rank}\,(\mathsf C(\Phi))}\to {\rm ker}\,(\mathsf C(\Phi))$ any unitary.
\item[(ii)] 
$
\Phi(|x^\perp\rangle\langle x^\perp|)\mathcal K_\Phi^\perp   
|x\rangle=\{0\}
$
\item[(iii)]
If $\{K_j\}_j$ is any set of Kraus operators of $\Phi$ and $\{G_k\}_k$ is any basis of $\mathcal K_\Phi^\perp$, then it holds that $\langle x^\perp|K_j^\dagger G_k|x\rangle=0$ for all $j,k$.
\end{itemize}
\end{proposition}
\begin{proof}
To show that condition (i) is necessary assume that $\Phi\Psi_\varepsilon^{-1}$ is completely positive for some $\varepsilon\in(0,1)$.
By Lemma~\ref{lemma_properties_Psilambda}~(iii) this implies that the same is true (i.e., $\mathsf C(\Phi\Psi_\lambda^{-1})\geq 0$) for all $\lambda\in[0,\varepsilon]$.

First we re-write Eq.~\eqref{eq:C_Phi_Psiinv}
in the following way:
Defining $Z:=|x^\perp\rangle\langle x|^T\otimes{\bf1}=|\overline{x}\rangle\langle\overline{x^\perp}|\otimes{\bf1}$---where $\overline{(\cdot)}$, here and henceforth, denotes the complex conjugate---as well as
$f:[0,1)\to[0,\infty)$ via $f(\lambda):=(1-\lambda)^{-1/2}-1$, one readily finds
$(K_\lambda^{-1})^T={\bf1}+f(\lambda)|\overline x\rangle\langle\overline x|$,
$(K_\lambda^{-1})^T\otimes{\bf1}={\bf1}+f(\lambda)ZZ^\dagger $,
and
$L_\lambda^T\otimes{\bf1}=\sqrt\lambda Z$.
Thus we obtain
\begin{align}
\mathsf C(\Phi\Psi_\lambda^{-1})&=({\bf1}+f(\lambda)ZZ^\dagger )\mathsf C(\Phi)({\bf1}+f(\lambda)ZZ^\dagger )-\frac{\lambda}{1-\lambda}Z\mathsf C(\Phi)Z^\dagger \notag\\
&= \mathsf C(\Phi)+f(\lambda)\{ZZ^\dagger ,\mathsf C(\Phi)\}+f(\lambda)^2ZZ^\dagger \mathsf C(\Phi)ZZ^\dagger -\frac{\lambda}{1-\lambda}Z\mathsf C(\Phi)Z^\dagger \,.\label{eq:PhiPsi_expansion}
\end{align}
Now let $\iota:\mathbb C^{n^2-{\rm rank}\,(\mathsf C(\Phi))}\to {\rm ker}\,(\mathsf C(\Phi))$ be any unitary. Because $ \mathsf C(\Phi\Psi_\lambda^{-1})\geq 0$ Eq.~\eqref{eq:PhiPsi_expansion} implies
$$\iota^\dagger \Big(\mathsf C(\Phi)+f(\lambda)\{ZZ^\dagger ,\mathsf C(\Phi)\}+f(\lambda)^2ZZ^\dagger \mathsf C(\Phi)ZZ^\dagger -\frac{\lambda}{1-\lambda}Z\mathsf C(\Phi)Z^\dagger \Big)\iota\geq 0\,.$$
Moreover, as the range of $\iota$ is the kernel of $\mathsf C(\Phi)$ this simplifies to
$$
\iota^\dagger \Big(f(\lambda)^2ZZ^\dagger \mathsf C(\Phi)ZZ^\dagger -\frac{\lambda}{1-\lambda}Z\mathsf C(\Phi)Z^\dagger \Big)\iota\geq 0\,.
$$
Recalling that this holds for all $\lambda\in[0,\varepsilon]$ 
and that $f(0)=0$,
the first derivative at zero has to be positive semi-definite:
Given a general curve $\{A(t)\}_{t\in[0,\varepsilon]}$, $\varepsilon>0$ of positive semi-definite matrices which satisfies $A(0)=0$ and which is
differentiable at zero, $A'(0)=\lim_{h\to 0^+}h^{-1}A(h)\geq 0$ because the positive semi-definite matrices form a closed cone.
Thus we find as a necessary condition
\begin{align*}
0
&\leq\iota^\dagger \Big(\frac{d}{d\lambda}f(\lambda)^2\Big|_{\lambda=0}ZZ^\dagger \mathsf C(\Phi)ZZ^\dagger -\frac{d}{d\lambda}\frac{\lambda}{1-\lambda}\Big|_{\lambda=0}Z\mathsf C(\Phi)Z^\dagger \Big)\iota=-\iota^\dagger Z\mathsf C(\Phi)Z^\dagger \iota\,,
\end{align*}
i.e., $\iota^\dagger Z\mathsf C(\Phi)Z^\dagger \iota\leq 0$;
here we used that $\frac{d}{d\lambda}f(\lambda)^2|_{\lambda=0}=2f'(0)f(0)=0$
and
$\frac{d}{d\lambda}\frac{\lambda}{1-\lambda}|_{\lambda=0}=1$.
On the other hand $\iota^\dagger Z\mathsf C(\Phi)Z^\dagger\iota\geq 0$ because $\mathsf C(\Phi)\geq 0$ so altogether we find $\iota^\dagger Z\mathsf C(\Phi)Z^\dagger \iota=0$, i.e., (i).

``(i) $\Rightarrow$ (ii)'': Let any $Z\in\mathcal K_\Phi^\perp$ be given; we have to show $\Phi(|x^\perp\rangle\langle x^\perp|)Z   
|x\rangle=0$. 
First note  ${\rm vec}(Z)\in{\rm ker}(\mathsf C(\Phi))$ by Lemma~\ref{lemma_kerC_vecKPhi} so
$\langle {\rm vec}(Z)|(|x^\perp\rangle\langle x|^T\otimes{\bf1})\mathsf C(\Phi)(|x^\perp\rangle\langle x|^T\otimes{\bf1})^\dagger |{\rm vec}(Z)\rangle=0$
by~(i).
Substituting the definitions of ${\rm vec}$ and $\mathsf C$, a straightforward computation shows that the latter expression is equal to
$\langle x|  Z^\dagger \Phi(|x^\perp\rangle\langle x^\perp|)Z| x\rangle=\|\sqrt{ \Phi(|x^\perp\rangle\langle x^\perp|)}Z| x\rangle\|^2$.
But we just showed that this is zero so 
$
 \Phi(|x^\perp\rangle\langle x^\perp|)Z| x\rangle=\sqrt{ \Phi(|x^\perp\rangle\langle x^\perp|)}\sqrt{ \Phi(|x^\perp\rangle\langle x^\perp|)}Z| x\rangle=0
$,
as desired.

``(ii) $\Rightarrow$ (iii)'': 
Condition~(ii) is linear in the considered element of $\mathcal K_\Phi^\perp$ meaning it suffices to check on a basis, i.e., $\Phi(|x^\perp\rangle\langle x^\perp|)G_k   
|x\rangle=0$ for all $k$. Using any Kraus operators $\{K_j\}_j$ of $\Phi$
$$
0=\langle x|G_k^\dagger \Phi(|x^\perp\rangle\langle x^\perp|)G_k|x\rangle=\sum_j\langle x|G_k^\dagger K_j|x^\perp\rangle\langle x^\perp|K_j^\dagger G_k |x\rangle=\sum_j|\langle x^\perp|K_j^\dagger G_k |x\rangle|^2
$$
for all $k$ which can only be true if $\langle x^\perp|K_j^\dagger G_k |x\rangle=0$ for all $j,k$.

``(iii) $\Rightarrow$ (i)'':  We have to show $\langle y|  (|x^\perp\rangle\langle x|^T\otimes{\bf1})\mathsf C(\Phi)(|x^\perp\rangle\langle x|^T\otimes{\bf1})^\dagger |z\rangle=0$ for any vectors $y,z\in{\rm ker}(\mathsf C(\Phi))$. Much like in the proof of ``(i) $\Rightarrow$ (ii)'', using Lemma~\ref{lemma_kerC_vecKPhi} one sees that this is equivalent to $\langle x|  {\rm vec}^{-1}(y)^\dagger \Phi(|x^\perp\rangle\langle x^\perp|){\rm vec}^{-1}(z)  | x\rangle=0$, resp.~$\langle x|  Y^\dagger\Phi(|x^\perp\rangle\langle x^\perp|)Z  | x\rangle=0$ for all $Y,Z\in\mathcal K_\Phi^\perp$.
But this evaluates to
$$
\langle x|  Y^\dagger\Phi(|x^\perp\rangle\langle x^\perp|)Z  | x\rangle=\sum_j\langle x|  Y^\dagger K_j|x^\perp\rangle\langle x^\perp|K_j^\dagger Z  | x\rangle=\sum_j\overline{\langle x^\perp|  K_j^\dagger Y|x\rangle}\langle x^\perp|K_j^\dagger Z  | x\rangle
$$
which by assumption (and linearity) is zero, as desired.
\end{proof}

\subsection{Sufficient criterion for complete positivity of $\Phi\Psi_\lambda^{-1}$}\label{sec32}

At this point we are ready to state and prove our main result which is a slight modification of the necessary condition from Proposition~\ref{prop_necess_division}:
on top of $\langle x^\perp|K_j^\dagger G_k|x\rangle$ vanishing for all $j,k$ we will require that the same is true for $\langle x|K_j^\dagger G_k|x\rangle$.
More precisely, the following result holds:

\begin{thm}\label{thm_main_div}
Let $\Phi\in\mathsf{CPTP}(n)$ and orthogonal unit vectors $x,x^\perp\in\mathbb C^n$ be given. Assume that one of the following equivalent statements holds:
\begin{itemize}
\item[(i)] $\Phi(|x^\perp\rangle\langle x^\perp|)\mathcal K_\Phi^\perp|x\rangle=\{0\}=\Phi(|x\rangle\langle x|)\mathcal K_\Phi^\perp   
|x\rangle$.
\item[(ii)] $\langle x^\perp|K_j^\dagger G_k|x\rangle=0=\langle x|K_j^\dagger G_k|x\rangle$ for all $j,k$, where $\{K_j\}_j$ is any set of Kraus operators of $\Phi$ and $\{G_k\}_k$ is any basis of $\mathcal K_\Phi^\perp$.
\end{itemize}
Then there exists $\lambda\in(0,1)$ such that $\Phi\Psi_\lambda^{-1}$ is completely positive.
In particular, $\Phi=(\Phi\Psi_\lambda^{-1})\Psi_\lambda$ is divisible.
Moreover if, in addition, $\Phi(|x\rangle\langle x|)\neq\Phi(|x^\perp\rangle\langle x^\perp|)$,
then there exists $\lambda\in(0,1)$ such that 
$\Phi\Psi_\lambda^{-1}$ is completely positive and its Kraus rank is strictly smaller than the Kraus rank of $\Phi$.
\end{thm}

\begin{proof}[Proof of Thm.~\ref{thm_main_div}]
First, equivalence of~(i) and~(ii) is shown just like
Prop.~\ref{prop_necess_division} ``(ii) $\Leftrightarrow$ (iii)''.

Now assume that~(ii) holds; we have to find $\lambda\in(0,1)$ such that $\mathsf C(\Phi\Psi_\lambda^{-1})\geq 0$.
Recall Eq.~\eqref{eq:PhiPsi_expansion}
$$
\mathsf C(\Phi\Psi_\lambda^{-1})= \mathsf C(\Phi)+f(\lambda)\{ZZ^\dagger ,\mathsf C(\Phi)\}+f(\lambda)^2ZZ^\dagger \mathsf C(\Phi)ZZ^\dagger -\frac{\lambda}{1-\lambda}Z\mathsf C(\Phi)Z^\dagger
$$
(as well as the notation $Z=|\overline{x}\rangle\langle\overline{x^\perp}|\otimes{\bf1}$)
and let $\iota:\mathbb C^{n^2-{\rm rank}\,(\mathsf C(\Phi))}\to {\rm ker}\,(\mathsf C(\Phi))$ be any unitary.
In Proposition~\ref{prop_necess_division} we saw that $\iota^\dagger Z\mathsf C(\Phi)Z^\dagger \iota=0$ is equivalent to $\langle x^\perp|X^\dagger K_j|x\rangle=0$ for all $j$ and all $X\in\mathcal K_\Phi^\perp$. Because $ZZ^\dagger=|\overline x\rangle\langle\overline x|\otimes{\bf1}$ one, analogously, finds that  $\iota^\dagger ZZ^\dagger \mathsf C(\Phi)ZZ^\dagger \iota=0$ is equivalent to $\langle x|X^\dagger K_j|x\rangle=0$ for all $j$ and all $X\in\mathcal K_\Phi^\perp$.
In other words~(ii) forces both $\iota^\dagger Z\mathsf C(\Phi)Z^\dagger \iota$ and $\iota^\dagger ZZ^\dagger \mathsf C(\Phi)ZZ^\dagger \iota$ to vanish.
But $\Phi\in\mathsf{CP}(n)$ implies
$
0=\iota^\dagger Z\mathsf C(\Phi)Z^\dagger \iota=( \sqrt{\mathsf C(\Phi)} Z^\dagger\iota )^\dagger( \sqrt{\mathsf C(\Phi)} Z^\dagger\iota )
$
which shows $ \mathsf C(\Phi)Z^\dagger\iota=\sqrt{\mathsf C(\Phi)}  \sqrt{\mathsf C(\Phi)} Z^\dagger\iota=0$; similarly, $\iota^\dagger ZZ^\dagger \mathsf C(\Phi)ZZ^\dagger \iota=0$ does imply $ \mathsf C(\Phi) ZZ^\dagger\iota=0$. Altogether, combining Eq.~\eqref{eq:PhiPsi_expansion} with our assumption~(ii) shows $\mathsf C(\Phi\Psi_\lambda^{-1})\iota=0$, that is, the kernel of $\mathsf C(\Phi)$ does not shrink when applying $\Psi_\lambda^{-1}$ for any $\lambda>0$.
This is the key to the desired complete positivity:
given any unitary $\xi:\mathbb C^{{\rm rank}\,(\mathsf C(\Phi))}\to {\rm ran}\,(\mathsf C(\Phi))$ (so $\xi\xi^\dagger +\iota\iota^\dagger ={\bf1}_{n^2}$)
\begin{equation}\label{eq:C_decomp_xi}
\mathsf C(\Phi\Psi_\lambda^{-1})=(\xi\xi^\dagger +\iota\iota^\dagger ) \mathsf C(\Phi\Psi_\lambda^{-1}) (\xi\xi^\dagger +\iota\iota^\dagger )= \xi\big(\xi^\dagger \mathsf C(\Phi\Psi_\lambda^{-1})\xi\big)\xi^\dagger \,.
\end{equation}
   But $\xi^\dagger \mathsf C(\Phi\Psi_\lambda^{-1})\xi\to\xi^\dagger \mathsf C(\Phi)\xi$ as $\lambda\to 0^+$ and the latter operator is positive definite (by definition of $\xi$).
Thus, because the eigenvalues of continuous paths are continuous \cite[Ch.~II, Thm.~5.2]{Kato80} there has to exist $\lambda>0$ such that $\xi^\dagger \mathsf C(\Phi\Psi_{\widehat\lambda}^{-1})\xi>0$ for all $\widehat\lambda\in[0,\lambda]$. In particular, this by Eq.~\eqref{eq:C_decomp_xi} means that $\mathsf C(\Phi\Psi_{\widehat\lambda}^{-1})\geq 0$ for all $\widehat\lambda\in[0,\lambda]$.
This also proves divisibility: Choose any 
$\widehat\lambda\in[0,\lambda]\cap(0,1)$.
If the Kraus rank of $\Phi\Psi_{\widehat\lambda}^{-1}$ is two or larger, then $\Phi=(\Phi\Psi_{\widehat\lambda}^{-1})(\Psi_{\widehat\lambda})$ is the product of two non-unitary channels. If, however, $\Phi\Psi_{\widehat\lambda}^{-1}$ is unitary, then by Lemma~\ref{lemma_properties_Psilambda}~(ii)
$$
\Phi=\Big(\Phi\Psi_{\widehat\lambda}^{-1}\Psi_{\frac{\widehat\lambda}{2-\widehat\lambda}}\Big)\Psi_{\frac{\widehat\lambda}2}\,,
$$
i.e., we expressed $\Phi$ as a product of channels of Kraus rank $2$ (due to $0<\frac{\widehat\lambda}{2}<\frac{\widehat\lambda}{2-\widehat\lambda}<1$).

Having established divisibility assume now that, in addition, $\Phi(|x\rangle\langle x|)\neq\Phi(|x^\perp\rangle\langle x^\perp|)$.
We will show that this guarantees the existence of $\mu\in(0,1)$ such that $\Phi\Psi_\mu^{-1}$ is not positive anymore.
This would conclude the proof: on the one hand, we already saw that $\xi^\dagger \mathsf C(\Phi\Psi_{\widehat\lambda}^{-1})\xi>0$ for all $\widehat\lambda\in[0,\lambda]$, and on the other hand $\mathsf C(\Phi\Psi_\mu^{-1})\not\geq 0$ does---by Eq.~\eqref{eq:C_decomp_xi}---force $\xi^\dagger \mathsf C(\Phi\Psi_{\mu}^{-1})\xi$ to have a negative eigenvalue.
So, again, by continuity of the eigenvalues\cite[Ch.~II, Thm.~5.2]{Kato80}
there exists $\lambda\in(0,\mu)\subset(0,1)$ such that all eigenvalues of $\xi^\dagger \mathsf C(\Phi\Psi_{\lambda}^{-1})\xi$ are non-negative but at least one of them is zero. Equivalently, the Kraus rank of $\mathsf C(\Phi\Psi_\lambda^{-1})= \xi\big(\xi^\dagger \mathsf C(\Phi\Psi_\lambda^{-1})\xi\big)\xi^\dagger $ is strictly smaller than the Kraus rank of $\mathsf C(\Phi)= \xi\big(\xi^\dagger \mathsf C(\Phi)\xi\big)\xi^\dagger $, as desired.

Now for the missing implication. Assume $\Phi(|x\rangle\langle x|)\neq\Phi(|x^\perp\rangle\langle x^\perp|)$. First note that this implies $\Phi(|x\rangle\langle x|)\not\geq\Phi(|x^\perp\rangle\langle x^\perp|)$: by contrapositive, if $\Phi(|x\rangle\langle x|)-\Phi(|x^\perp\rangle\langle x^\perp|)\geq 0$, then 
\begin{align*}
0=1-1={\rm tr}\big( \Phi(|x\rangle\langle x|)-\Phi(|x^\perp\rangle\langle x^\perp|)\big)=\big\|\Phi(|x\rangle\langle x|)-\Phi(|x^\perp\rangle\langle x^\perp|)\big\|_1
\end{align*}
(with $\|A\|_1:={\rm tr}(\sqrt{A^\dagger A})$ the usual trace norm)
and thus $\Phi(|x\rangle\langle x|)=\Phi(|x^\perp\rangle\langle x^\perp|)$.
In particular, there has to exist $y\in\mathbb C^n$ such that $\langle y|(\Phi(|x\rangle\langle x|)-\Phi(|x^\perp\rangle\langle x^\perp|))|y\rangle<0$.
On the other hand, Eq.~\eqref{eq:Psiinv_xy_orthogonal} implies
$$\Psi_\lambda^{-1}(|x\rangle\langle x|)=\frac1{1-\lambda} |x\rangle\langle x|-\Big(\frac{\lambda}{1-\lambda}\Big)|x^\perp\rangle\langle x^\perp|=\frac1{1-\lambda}(|x\rangle\langle x|-|x^\perp\rangle\langle x^\perp|)+|x^\perp\rangle\langle x^\perp|$$
for all $\lambda\in[0,1)$.
Combining these facts yields
$$
\big\langle y\big| \Phi\Psi_\lambda^{-1}(  |x\rangle\langle x|  )   \big|y\big\rangle=
\frac1{1-\lambda}\langle y|\Phi(|x\rangle\langle x|)-\Phi(|x^\perp\rangle\langle x^\perp|)|y\rangle+\langle y|\Phi(|x^\perp\rangle\langle x^\perp|)|y\rangle\,.
$$
But in the limit $\lambda\to 1^-$ this goes to $-\infty$ (because $\langle y|(\Phi(|x\rangle\langle x|)-\Phi(|x^\perp\rangle\langle x^\perp|))|y\rangle<0$) so there has to exist $\mu\in(0,1)$ for which $\Phi\Psi_\mu^{-1}$ is not positive anymore.
\end{proof}

\noindent This sufficient condition of ours is similar in spirit to recent divisibility results of Davalos and Ziman \cite{DZ23} obtained in the context of open systems and Markovian dynamics.
Note that repeated application of Theorem~\ref{thm_main_div} may be possible, that is, there may exist $m>1$, $\lambda_1,\ldots,\lambda_m\in(0,1)$ and pairs of orthogonal unit vectors $(x_1,x_1^\perp),\ldots,(x_m,x_m^\perp)\in\mathbb C^n\times\mathbb C^n$ such that
\begin{equation}\label{eq:Phi_repeat}
\Phi=\big(\Phi\Psi_{\lambda_m,x_m,x_m^\perp}^{-1}\ldots\Psi_{\lambda_1,x_1,x_1^\perp}^{-1}\big)\Psi_{\lambda_1,x_1,x_1^\perp}\ldots\Psi_{\lambda_m,x_m,x_m^\perp}
\end{equation}
and the difference between the Kraus rank of $\Phi$ and the Kraus rank of $\Phi\Psi_{\lambda_m,x_m,x_m^\perp}^{-1}\ldots\Psi_{\lambda_1,x_1,x_1^\perp}^{-1}$ is at least $m$.
This procedure has to terminate eventually because at some point $\Phi\Psi_{\lambda_m,x_m,x_m^\perp}^{-1}\ldots\Psi_{\lambda_1,x_1,x_1^\perp}^{-1}$ ($m\geq 1$) is either of the form $\Psi_{\lambda,x,x^\perp}$ itself or no such channel can be factored out from it anymore.
Be aware that the latter case does not say anything about whether that resulting channel is divisible, just that factoring out the particular elementary building block we chose is not possible anymore.\medskip

With our main result established let us present some corollaries which feature easier-to-check special cases:

\begin{corollary}\label{coro_Kphix}
Given any $\Phi\in\mathsf{CPTP}(n)$,
if there exists $x\neq 0$ such that $\mathcal K_\Phi^\perp   
|x\rangle=0$, then $\Phi$ is divisible.
\end{corollary}
This also reproduces the result mentioned in the introduction that every channel $\Phi$ with maximal Kraus rank is divisible \cite[Thm.~11]{Wolf08a} because then, trivially, $\mathcal K_\Phi^\perp=\{0\}$ (Lemma~\ref{lemma_kerC_vecKPhi}).
\begin{remark}
Note that Coro.~\ref{coro_Kphix} is equivalent to $\bigcap_{k}{\rm ker}(G_k)\neq\{0\}$ for any basis $\{G_k\}_k$ of $\mathcal K_\Phi^\perp$;
in particular, this implies $\bigcap_{j,k}{\rm ker}(K_j^\dagger G_k)\neq\{0\}$ for any set $\{K_j\}_j$ of Kraus operators of $\Phi$ which by Thm.~\ref{thm_main_div}~(ii) is sufficient for divisibility.
However this kernel condition is far from necessary, unless $n=2$: an equivalent formulation of Thm.~\ref{thm_main_div}~(ii) is that there exists an orthonormal basis in which all $\{K_j^\dagger G_k\}_{j,k}\subset\mathbb C^{n\times n}$ have a zero in the top-left corner and a zero just below that. In the single-qubit case---because any orthogonal unit vectors $x,x^\perp\in\mathbb C^2$ form an orthonormal basis of $\mathbb C^2$---this boils down to $\bigcap_{j,k}{\rm ker}(K_j^\dagger G_k)\neq\{0\}$.
In particular, this yields a simple-to-verify sufficient condition for divisibility specifically for qubit channels.
\end{remark}

It should not come as a surprise that our result reproduces the sufficient condition for divisibility of a non-negative matrix $A$ by Richman and Schneider 
\cite[Thm~2.4]{RS74}
mentioned in the introduction: Their condition (${\rm sgn}(A)|j\rangle\geq {\rm sgn}(A)|k\rangle$ for some $j\neq k$) translates to a condition for classical channels via the standard embedding $\Phi_A(|j\rangle\langle k|):=\delta_{jk}\sum_{l=1}^nA_{lj}|l\rangle\langle l|$
which is equivalent to our necessary (Prop.~\ref{prop_necess_division}) as well as our sufficient (Thm.~\ref{thm_main_div}) condition.\medskip

While in this classical case the Choi matrix is always diagonal, interestingly, our main result can be applied to block-diagonal cases (under local unitaries) assuming the blocks are of certain sizes:

\begin{corollary}\label{coro_div_blockdiag}
Let $\Phi\in\mathsf{CPTP}(n)$ and $U,V\in\mathsf{SU}(n)$ be given such that one of the following holds:
\begin{itemize}
\item[(i)] There exist $m\in\{n,\ldots,n^2\}$ as well as $X_1\in\mathbb C^{m\times m}$, $X_2\in\mathbb C^{(n^2-m)\times(n^2-m)}$ both positive semi-definite such that
$$
\mathsf C(\Phi)=(U^T\otimes V^\dagger )(X_1\oplus X_2)(U^T\otimes V^\dagger )^\dagger \quad
\text{ or }\quad
\mathsf C(\Phi)=(U^T\otimes V^\dagger )(X_2\oplus X_1)(U^T\otimes V^\dagger )^\dagger \,.
$$
\item[(ii)] There exist $m\in\{2n-1,\ldots,n^2\},m_1\in\{0,\ldots,n^2-m\}$ as well as $X_0\in\mathbb C^{m_1\times m_1}$, $X_1\in\mathbb C^{m\times m}$, $X_2\in\mathbb C^{(n^2-m-m_1)\times(n^2-m-m_1)}$ each positive semi-definite such that
$$
\mathsf C(\Phi)=(U^T\otimes V^\dagger )(X_0\oplus X_1\oplus X_2)(U^T\otimes V^\dagger )^\dagger \,.
$$
\end{itemize}
If $X_1>0$, then there exist  $\lambda\in(0,1)$ and orthogonal unit vectors $x,x^\perp\in\mathbb C^n$ such that
$\Phi\Psi_{\lambda}^{-1}$ is completely positive.
In particular, $\Phi=(\Phi\Psi_{\lambda}^{-1})\Psi_{\lambda}$ is divisible in this case.

\end{corollary}
\begin{proof}
Without loss of generality assume $U=V={\bf1}$ (due to Lemma~\ref{lemma_div_unitary_equiv} combined with Lemma~\ref{lemma_choi_prod}).
In the first case $\mathsf C(\Phi)=X_1\oplus X_2$ implies
$\ker(\mathsf C(\Phi))=\ker(X_1)\oplus\ker(X_2)$.
But $X_1>0$ by assumption so Lemma~\ref{lemma_kerC_vecKPhi} shows that the first column of every element of ${\rm vec}^{-1}({\rm ker}(\mathsf C(\Phi)))=\mathcal K_\Phi^\perp$ is zero.
Hence $\mathcal K_\Phi^\perp|0\rangle=\{0\}$ which yields divisibility by Coro.~\ref{coro_Kphix}.
The cases $\mathsf C(\Phi)=X_2\oplus X_1$ and $\mathsf C(\Phi)=X_0\oplus X_1\oplus X_2$ are shown analogously: Due to the size of $X_1$
there always exists one column which is always zero for every element of $\mathcal K_\Phi^\perp$, so Coro.~\ref{coro_Kphix} can be applied accordingly.
\end{proof}
\noindent Note that without further assumptions on the sizes of $X_0,X_2$ in Coro.~\ref{coro_div_blockdiag}~(ii) the bound $m\geq 2n-1$ is optimal in general. For this consider
$$
\mathsf C(\Phi)=\begin{pmatrix}
0&0&0&0\\
0&1&0&0\\
0&0&1&0\\
0&0&0&0
\end{pmatrix}
$$
so $m=2<2n-1$ but $\mathcal K_\Phi^\perp={\rm vec}^{-1}({\rm ker}(\mathsf C(\Phi)))={\rm span}\{|0\rangle\langle 0|,|1\rangle\langle 1|\}$ does not admit any $x\neq 0$ such that $\mathcal K_\Phi^\perp|x\rangle=0$.
However, for suitable $m_1$ one can find lower values of $m$ which make the construction work (e.g., if $m_1$ is itself a multiple of $n$, then all $m\geq n$ work).

\section{Examples}\label{sec_ex}

Before concluding let us look at some examples which will, hopefully, illustrate
the scope of our divisibility criteria.
First, as mentioned in the introduction,
unital qubit channels of Kraus rank $3$ are the only indivisible qubit channels.
Thus our sufficient criterion necessarily fails, and in this first example we want to understand the underlying reason.

\begin{example}
Let $\Phi\in\mathsf{CPTP}(2)$ be unital with Kraus rank 3; in this example we will see that $\Phi\Psi_{\lambda}^{-1}$ is never positive, regardless of how $x,x^\perp$, and $\lambda\in(0,1)$ are chosen. First, it is well known that because $\Phi$ is unital there exist $U,V\in\mathsf U(2)$ such that the Pauli transfer matrix ($(\frac12{\rm tr}(\sigma_j\Phi(\sigma_k)))_{j,k}$) of $\Phi_{U,V}:={\rm Ad}_{V^\dagger }\cdot\Phi\cdot {\rm Ad}_U$ reads ${\rm diag}(1,x_1,x_2,x_3)$ with $1\geq x_1\geq x_2\geq|x_3|
$, cf.~\cite[Sec.~VI]{Wolf08a} or \cite{RSW02}.
This $\Phi_{U,V}$ has Kraus rank $3$ if and only if
there exist $0\leq b<a<1$
such that
$$
\mathsf C(\Phi_{U,V})=\begin{pmatrix}
1-a&0&0&1-a\\0&a&b&0\\0&b&a&0\\1-a&0&0&1-a
\end{pmatrix}\,.
$$
(Recall that basis transformations $U,V$ do not change the problem, cf.~Lemma~\ref{lemma_div_unitary_equiv}).
One readily verifies that $\Phi_{U,V}$ has Kraus operators $\{\sqrt{1-a}\,{\bf1},$
$\sqrt{(a+b)/2}\,\sigma_x,\sqrt{(a-b)/2}\,\sigma_y\}$) and that $\mathcal K_{\Phi_{U,V}}^\perp={\rm span}\{\sigma_z\}$.
This is why not just our sufficient, but already our necessary criterion (Prop.~\ref{prop_necess_division}) fails: If there existed orthogonal unit vectors $x,x^\perp\in\mathbb C^2$ such that $\langle x^\perp|K_j^\dagger \sigma_z|x\rangle=0$ for all $j=1,2,3$, then
$\{\sigma_z,\sigma_x\sigma_z,\sigma_y\sigma_z\}$ were simultaneously unitarily (upper) triangularizable which is impossible because
$\sigma_x,\sigma_y$
do not share a common eigenvector.
\end{example}

Another deep insight of Richman and Schneider \cite[Sec.~3]{RS74} was that a non-negative matrix $A\in\mathbb R_+^{3\times 3}$ is prime if and only if there exist permutation matrices $\tau,\pi$ such that $\tau A\pi$ has positive off-diagonals and the diagonal is all zeros.
Interestingly, this result has an analogue for qutrit channels $\Phi\in\mathsf{CPTP}(3)$: As we will see in the next example, our necessary criterion (Prop.~\ref{prop_necess_division}) implies that \textit{if} such a prime can be found on the diagonal of the Choi matrix of $\Phi$, then there is no way to divide $\Psi_\lambda$ from $\Phi$.

\begin{example}
Let $\Phi\in\mathsf{CPTP}(3)$ be given such that the diagonal of its Choi matrix $\mathsf C(\Phi)$ is given by $(0,a,1-a,b,0,1-b,c,1-c,0)$ for some $a,b,c\in(0,1)$ (we can disregard any possible permutations thanks to Lemma~\ref{lemma_div_unitary_equiv}). What we want to see now is that no matter how one chooses $x,x^\perp$ and no matter the off-diagonals of $\mathsf C(\Phi)$, $\Phi\Psi_\lambda^{-1}$ is never completely positive for any $\lambda\in(0,1)$.
First, because
$\mathsf C(\Phi)\geq 0$
the diagonal zeros force the following form:
$$
\mathsf C(\Phi)=\begin{pmatrix}
0&0&0&0&0&0&0&0&0\\
 0 & a & * & * & 0 & * & * & * & 0 \\
 0 & * & 1-a & * & 0 & * & * & * & 0 \\
 0 & * & * & b & 0 & * & * & * &  0\\
0&0&0&0&0&0&0&0&0\\
 0 & * & * & * & 0 & 1-b & * & * & 0 \\
 0 & * & * & * & 0 &  *& c & * & 0 \\
 0 & * & * & * & 0 & * & * & 1-c & 0 \\
0&0&0&0&0&0&0&0&0
\end{pmatrix}
$$
In particular Lemma~\ref{lemma_kerC_vecKPhi} shows that all $|j\rangle\langle j|$ are in $\mathcal K_\Phi^\perp$.
Thus by Prop.~\ref{prop_necess_division}~(ii),
for complete positivity of $\Phi\Psi_\lambda^{-1}$
it is necessary that $\Phi(|x^\perp\rangle\langle x^\perp|)|j\rangle\langle j|x\rangle=0$ for some orthogonal unit vectors $x,x^\perp$ and all $j$.
Next, trace-preservation of $\Phi$ forces some additional entries to vanish (marked in blue):
$$
\mathsf C(\Phi)=\begin{pmatrix}
0&0&0&0&0&0&0&0&0\\
 0 & a & * & * & 0 & * & * & {\color{blue}0} & 0 \\
 0 & * & 1-a & * & 0 & {\color{blue}0} & * & * & 0 \\
 0 & * & * & b & 0 & * & {\color{blue}0} & * &  0\\
0&0&0&0&0&0&0&0&0\\
 0 & * & {\color{blue}0} & * & 0 & 1-b & * & * & 0 \\
 0 & * & * & {\color{blue}0} & 0 &  *& c & * & 0 \\
 0 & {\color{blue}0} & * & * & 0 & * & * & 1-c & 0 \\
0&0&0&0&0&0&0&0&0
\end{pmatrix}
$$
Therefore $\Phi$ acts on any $X\in\mathbb C^{3\times 3}$ as follows:
$$
\Phi(X)=\begin{pmatrix}
bX_{22}+cX_{33}&*&*\\
*&aX_{11}+(1-c)X_{33}&*\\
*&*&(1-a)X_{11}+(1-b)X_{22}
\end{pmatrix}
$$
Inserting this into the necessary conditions $\langle j|\Phi(|x^\perp\rangle\langle x^\perp|)|j\rangle\langle j|x\rangle=0$, $j=1,2,3$ yields
\begin{align*}
\big(b\big|(x^\perp)_2\big|^2+c\big|(x^\perp)_3\big|^2\big) x_1 &=0\\
\big(a\big|(x^\perp)_1\big|^2+(1-c)\big|(x^\perp)_3\big|^2\big) x_2 &=0\\
\big((1-a)\big|(x^\perp)_1\big|^2+(1-b)\big|(x^\perp)_2\big|^2\big) x_3 &=0\,.
\end{align*}
But $x\neq 0$, meaning there exists $j$ such that---by the $j$-th condition---$x^\perp=c|j\rangle$
(as $a,b,c\in(0,1)$).
Hence $0=\langle x|x^\perp\rangle=cx_j$ shows $c=0$, i.e., $x^\perp=0$; this contradicts $x^\perp$ being a unit vector.
Reformulated,
the necessary condition from Prop.~\ref{prop_necess_division}~(ii) cannot be satisfied so $\Phi\Psi_\lambda^{-1}\not\in\mathsf{CP}(n)$ for all $\lambda>0$.
\end{example}
Finally, let us consider a constructive example, which will also serve as a comparison with the results of Wolf and Cirac on divisibility of qubit channels \cite[Sec.~VI]{Wolf08b}.
However, we want to stress that our divisibility criteria are more general as they work in arbitrary (finite) dimensions.

\begin{example}
Consider $\Phi\in\mathsf{CPTP}(2)$ defined via
$$
\mathsf C(\Phi)=\begin{pmatrix}
 1 & 0 & 0 & \frac{1}{3} \\
 0 & 0 & 0 & 0 \\
 0 & 0 & \frac{2}{3} & \frac{1}{3} \\
 \frac{1}{3} & 0 & \frac{1}{3} & \frac{1}{3} 
\end{pmatrix}.
$$
This channel has Kraus rank $3$ and is non-unital, meaning it can be written as a product of a non-unital channel of Kraus rank $2$ and a partial dephasing channel (in some basis) \cite[Thm.~18 ff.]{Wolf08a}:
$$
\widehat\Phi=\begin{pmatrix}
1&0&0&2/3\\0&1/3&0&1/3\\0&0&1/3&1/3\\0&0&0&1/3
\end{pmatrix}=\underbrace{\begin{pmatrix}
1&0&0&2/3\\0&1/\sqrt6&0&1/3\\0&0&1/\sqrt6&1/3\\0&0&0&1/3
\end{pmatrix}}_{=:\widehat{\Phi_1}}\underbrace{\begin{pmatrix}
1&0&0&0\\0&\sqrt{2/3}&0&0\\0&0&\sqrt{2/3}&0\\0&0&0&1
\end{pmatrix}}_{=:\widehat{\Phi_2}}
$$
Here, given arbitrary $\Psi\in\mathcal L(\mathbb C^{n\times n})$ we will write $\widehat\Psi$ for its representation matrix, i.e., $\widehat\Psi\in\mathbb C^{n^2\times n^2}$ is the unique matrix which satisfies ${\rm vec}(\Psi(X))=\widehat\Psi{\rm vec}(X)$ for all $X\in\mathbb C^{n\times n}$.
With this let us focus on applying our division algorithm to $\Phi$; for this we will use Coro.~\ref{coro_div_blockdiag}. Because
$$
\mathsf C(\Phi)=(\sigma_x\otimes{\bf1})\begin{pmatrix}
 \frac{2}{3} & \frac{1}{3} & 0 & 0 \\
 \frac{1}{3} & \frac{1}{3} & \frac{1}{3} & 0 \\
 0 & \frac{1}{3} & 1 & 0 \\
 0 & 0 & 0 & 0
\end{pmatrix}(\sigma_x\otimes{\bf1})^\dagger 
$$
and because
the top-right $3\times 3$ block is positive definite we can apply
Coro.~\ref{coro_div_blockdiag}~(i) with $m=3$, $X_2=0$, $x=\sigma_x|0\rangle=|1\rangle$, and $x^\perp=|0\rangle$ (recall Lemma~\ref{lemma_div_unitary_equiv}).
Thus $\Psi_\lambda$ is defined via the Kraus operators
$$
\begin{pmatrix}
1&0\\0&\sqrt{1-\lambda}
\end{pmatrix},\begin{pmatrix}
0&\sqrt\lambda\\0&0
\end{pmatrix}
$$
for all $\lambda\in[0,1]$
meaning
$$
\mathsf C(\Phi\Psi_\lambda^{-1})=\begin{pmatrix}
1&0&0&\frac{1}{3\sqrt{1-\lambda}}\\
0&0&0&0\\
0&0&\frac{2-3\lambda}{3-3\lambda}&\frac{1}{3-3\lambda}\\
\frac{1}{3\sqrt{1-\lambda}}&0&\frac{1}{3-3\lambda}&\frac{1}{3-3\lambda}
\end{pmatrix}.
$$
It is not difficult to see that $\mathsf C(\Phi\Psi_\lambda^{-1})\geq 0$ if and only if $\lambda\in[0,\frac16]$, and that $\Phi\Psi_{1/6}^{-1}$ is the channel with Kraus rank $2$ we are looking for (the existence of which is guaranteed by Thm.~\ref{thm_main_div} because
$\Phi(|0\rangle\langle 0|)\neq\Phi(|1\rangle\langle 1|)$).
This yields the decomposition of $\Phi$ into two (non-unital) channels both of Kraus rank $2$ via
$$
\widehat\Phi=\underbrace{\begin{pmatrix}
1&0&0&3/5\\
0&\sqrt{2/15}&0&2/5\\
0&0&\sqrt{2/15}&2/5\\
0&0&0&2/5
\end{pmatrix}}_{=\widehat{\Phi\Psi_\lambda^{-1}}}\underbrace{\begin{pmatrix}
1&0&0&1/6\\
0&\sqrt{5/6}&0&0\\
0&0&\sqrt{5/6}&0\\
0&0&0&5/6
\end{pmatrix}}_{=\widehat{\Psi_\lambda}}.
$$
\end{example}

\section{Conclusions \& Outlook}\label{sec_coout}

In this work we contributed to the common understanding of divisibility, a concept central to open systems theory, quantum dynamics, as well as quantum systems engineering. More precisely, we extended an idea from divisibility of non-negative matrices to the quantum realm which resulted in a simple algorithm to divide a given quantum channel $\Phi$ by certain ``elementary'' channels $\Psi_\lambda$. 
We found a necessary (Prop.~\ref{prop_necess_division}) as well as a sufficient condition (Thm.~\ref{thm_main_div}) for when this division $\Phi\Psi_\lambda^{-1}$ results in a valid channel, both relying on the Kraus subspace and its orthogonal complement.
This is the first non-trivial divisibility criterion which is applicable beyond qubits. Generically, this division even lowers the Kraus rank which is why repeated application of this algorithm (if possible) results in a factorization of $\Phi$ into in some sense ``simple'' channels (Eq.~\eqref{eq:Phi_repeat}).

Yet, the tools we employed are not limited to the elementary channel $\Psi_\lambda$ we chose.
Indeed, this work can be seen as a blueprint which can, in principle, be applied to any scenario of this form and which yields necessary and sufficient conditions for divisibility w.r.t.~the chosen $\Psi_\lambda$ (although the actual computations may become more involved).
Thus, in a way, our work can be adapted to tackle the question posed in the introduction where an experimenter has access to a set of channels $\mathcal S$ and they want to know whether some target channel $\Phi$ can be factorized using elements from $\mathcal S$ (and if so, how to do it).

An open question at this point is whether our sufficient condition is also necessary, and whether our necessary condition is also sufficient.
We did comment on how for classical channels $\Phi_A$ these conditions are equivalent (although the underlying reason for this fails in general).
A different
scenario where one could tackle this question are covariant channels, i.e., $[\Phi,[H,\cdot]]=0$ for some non-trivial Hamiltonian $H$:
in this case $[\mathsf C(\Phi),H^T\otimes{\bf1}-{\bf1}\otimes H]=0$ as is readily verified, making $\mathsf C(\Phi)$ a lot more structured for analysis purposes.
Another follow-up could be to (try to) use our results to extend the universal set of qubit channels (i.e., a set for which the generated semigroup are all channels)
\cite[Thm.~5.19]{BGN14}
to higher dimensions.
However, already the classical case of non-negative matrices shows that this requires knowledge of the indivisible channels in three and more dimensions.
While the
simpler, classical question has not been investigated in detail yet
it would be an interesting first step, and this paper's results should be useful for that.

\begin{acknowledgments}
I am grateful to Fereshte Shahbeigi for constructive comments during the preparation of this paper. This work has been supported by the Einstein Foundation (Einstein Research Unit on Quantum Devices) and the MATH+ Cluster of Excellence.
\end{acknowledgments}

\section*{Author Declarations}

\subsection*{Conflict of Interest}
The author has no conflicts to disclose.
\subsection*{Author Contributions}
\noindent\textbf{Frederik vom Ende:} Conceptualization (equal); Formal analysis (equal); Investigation (equal); Methodology (equal); Writing – original draft (equal); Writing – review and editing (equal).
\section*{Data Availability}
Data sharing is not applicable to this article as no new data were created or analyzed in this study.

%
%

\bibliography{../../../../../control21vJan20}

\end{document}